%% file: main.tex
\RequirePackage{mathtools}

\newif{\ifcs}
\cstrue%

\documentclass[runningheads]{llncs}

\usepackage{graphicx}
\input{packages.tex}
\ifcs
\input{macros-cs.tex}
\else
\input{macros.tex}
\fi

\begin{document}
\title{A note on commutative Kleene algebra}
%
%
\author{Paul Brunet}
\authorrunning{P. Brunet}
%
\institute{%
  University College London\\
  \url{paul.brunet-zamansky.fr}\\
  \email{paul@brunet-zamansky.fr}
}

%
\maketitle              
\begin{abstract}
  In this paper we present a detailed proof of an important result of
  algebraic logic: namely that the free commutative Kleene algebra is
  the space of semilinear sets.
  The first proof of this result was proposed by Redko in 1964, and
  simplified and corrected by Pilling in his 1970 thesis.
  However, we feel that a new account of this proof is needed
  now. This result has acquired a particular importance in recent
  years, since it is a key component in the completeness proofs of
  several algebraic models of concurrent computations (bi-Kleene
  algebra, concurrent Kleene algebra...). 
  To that effect, we present a new proof of this result.
\end{abstract}

\section{Introduction}
\label{sec:intro}

In this paper we present a detailed proof of an important result of
algebraic logic: namely that the free commutative Kleene algebra is
the space of semilinear sets. This theorem is of central importance,
in particular because it is necessary to obtain the completeness of
concurrent variants of Kleene algebra, e.g. bi-Kleene
algebras~\cite{laurenceCompletenessTheoremsBiKleene2014}, concurrent
Kleene
algebras~\cite{hoareConcurrentKleeneAlgebra2009,kappeConcurrentKleeneAlgebra2018},
and the recently introduced ``concurrent Kleene algebras with
observations''~\cite{ckao}.

According to Daniel Krob~\cite{krobCompleteSystemBrational1990}: ``a
theorem of Redko from 1964
(see~\cite{redkoAlgebraCommutativeEvents1964}), whose proof was
simplified and corrected by Pilling
(see~\cite{pillingAlgebraOperatorsRegular1970}), gives a complete
identities system for the commutative rational expressions''. An
account of Pilling's proof was also included in Conway's 1971
book~\cite[Chapter 11]{conwayRegularAlgebraFinite2012}.

However we feel that an accessible proof of this result is missing
from this picture. To our knowledge, Redko's original proof, published
in Russian, has not been translated to English. Pilling and Conway's
proofs suffer from another drawback: these were done and published
before the theory of Kleene algebra was settled. Since then, basic
definitions and notations have diverged enough to render their
text difficult to read by contemporary mathematicians. 

In particular the axiomatisation that both Redko and Pilling prove
complete differ from the one used in
e.g.~\cite{laurenceCompletenessTheoremsBiKleene2014,kappeConcurrentKleeneAlgebra2018}.
Indeed, they both rely on infinite axiom schemes, namely for each
$k>0$ they include an identity:
\begin{equation}
  \titerate e\equiv \paren{\un\tjoin\scale 1 e\tjoin\cdots\tjoin\scale{\paren{k-1}}e}\tadd\ptiterate{\scale k e}.\eqproof{star-continuity}
\end{equation}
This principle may be understood as a limited form of counting modulo
$k$: every natural number may be written as the addition of a number
below $k$ and a factor of $k$. By contrast, we avoid the need for an
infinite axiomatisation by relying on an inference rule, in the style
of
e.g.~\cite{kozenCompletenessTheoremKleene1994,laurenceCompletenessTheoremsBiKleene2014,kappeConcurrentKleeneAlgebra2018}:
$$e\tadd x\leqq x\Rightarrow \titerate e\tadd x\leqq x.$$
We do need the family of identities~\eqref{eq:star-continuity} in the
proof. However, instead of postulating theses as axioms, we show that
they may be derived from our finitary axiomatisation.

Besides showing that the axiomatisation used by Redko and Pilling can
be derived from the more ``standard'' forms found in the literature,
we give a step-by-step description of the proof, highlighting the
techniques used and the key intermediary results that are needed for
this proof. We strive to use standard definitions as much as possible,
and to emphasise the relative difficulty of the various proof
steps. In our opinion, this proof could be used as a template for
formalisation in proof assistants, e.g. Coq or Isabelle.

The proof we present here loosely follows the strategy from Pilling's
PhD thesis, although we simplify some arguments and prove others in
more details.


The rest of the paper is organised as follows. In
Section~\ref{sec:def}, we lay down some definitions, and provide an
overview of the proof, identifying three main proof obligations. We
then devote some pages in Section~\ref{sec:prelim} to prove some
preliminary results. Finally, in
Sections~\ref{sec:decomp},~\ref{sec:incl-bases},
and~\ref{sec:inter-diff} we discharge the three remaining proof
obligations, thus finishing the proof of the main result.

\section{Definitions and overview of the proof}
\label{sec:def}

\subsection{Semi-linear sets}
\label{sec:semi-linear-sets}
A \emph{Parikh vector} is a $\alphabet$-indexed vector of natural
numbers. It can also be understood as a \emph{multiset} (i.e. a set
with multiplicity), or a \emph{commutative word} (i.e. a sequence
without order). 
The space of Parikh vectors is written $\Vectors$.
The vector with uniformly $0$ coordinate is written $\unitvec$.
Given a letter $a\in\alphabet$, we write $\letvec a$ for the vector with $1$ in coordinate $a$ and $0$ on every other coordinate.
The addition of two vectors $\vu,\vv\in\Vectors$, and the scalar multiplication of a number $n\in\Nat$ with a vector $\vu$ are defined coordinate-wise as usual:
\begin{mathpar}
  \vu\vecplus\vv\eqdef\ituple{\vu[a]+\vv[a]}{a\in\alphabet}\and
  \scale n \vu \eqdef\ituple{n\times\vu[a]}{a\in\alphabet}.
\end{mathpar}

Let $\Base=\set{\vec v_1,\cdots,\vec v_n}\subseteq\Vectors$ be a finite set of vectors, we call a $\Base$-point a vector $\vec\alpha\in\Points\eqdef\Nat^n$.  Such points can be interpreted as vectors by the function $\tovec\_:\Points\to\Vectors$ defined by
$\tovec\alpha\eqdef\vsum_{1\leqslant i\leqslant n}\scale{\alpha_i} {\vec v_i}$.

We now define the \emph{regular operators} on sets of vectors:
\begin{mathpar}
  \vU\vecplus\vV\eqdef\setcompr{\vu\vecplus\vv}{\vu\in\vU,\,\vv\in\vV}\and
  \vU\join\vV\eqdef\setcompr{\vu}{\vu\in\vU~or~\vv\in\vV}\and
  \iterate\vU\eqdef\setcompr{\vsum_{b\in\vU}\scale{\alpha_b} b}{\forall b\in\vU,\, \alpha_b\in\Nat}.
\end{mathpar}
Notice in particular that for a finite set $\Base$ we have $\iterate \Base=\setcompr{\tovec \alpha}{\alpha\in\Points}$.

The \emph{linear set generated by the vector $\vu$ and the finite set
  $\Base\subseteq\Vectors$} is defined by the following expression:
$$
\set\vu\vecplus\iterate\Base=\setcompr{\vu\vecplus\vsum_{b\in\Base}\scale{\alpha_b} b}{\forall b\in\Base, \alpha_b\in\Nat}.$$
A \emph{semilinear set} is a finite union of linear sets.


A finite set of vectors $\Base$ is called \emph{independent} if every
vector in $\Vectors$ has at most one decomposition in terms of the
vectors in $\Base$. In other words $\Base$ is independent iff
$\tovec \_$ is injective.

\subsection{Terms \& axioms}
\label{sec:terms}
Let $a,b,\dots\in\alphabet$ be a finite alphabet.
A \emph{\commutative regular expression} is a term generated by the following
grammar:
$$e,f\in\reg\Coloneqq \zero \Mid\un\Mid a\Mid e\tadd f\Mid e \tjoin f\Mid \titerate e.$$

Expressions may be immediately interpreted as sets of vectors:
\begin{align*}
  \sem \zero&\eqdef\emptyset
  &\sem\un&\eqdef\set\unitvec
  &\sem a&\eqdef\set{\letvec a}\\
  \sem {\titerate e}&\eqdef\iterate{\sem e}
  &\sem{e\tadd f}&\eqdef\sem e \vecplus\sem f
  &\sem {e\tjoin f}&\eqdef\sem e\join \sem f
\end{align*}

Notice that for any vector $\vec v\in\Vectors$, we may build an
expression $\expr {\vec v}\in\reg$ such that
$\sem{\expr{\vec v}}=\set{\vec v}$:
\begin{equation*}
\expr {\vec v}\eqdef
\tsum_{a\in\alphabet}\paren{\tsum_{1\leqslant j\leqslant \vec v_{a}} a}
=\begin{array}[t]{ccc}
  \underbrace{a\tadd\dots \tadd a}&\tadd\dots\tadd&\underbrace{z\tadd\dots\tadd z}\\
  \times\vec v_{a}&&\times\vec v_{z}
\end{array}.
\end{equation*}
Therefore, any semi-linear set may be represented as the semantics of
some \commutative regular expression.

\begin{table}[t]
  \centering
  \noindent%
  \begin{tableequations}
    \begin{minipage}[t]{.5\linewidth}
      \begin{align}
        e \tadd (f \tadd g) &\equiv (e \tadd f) \tadd g \axlabel{cka-seq-ass}\\
        e \tadd f &\equiv f \tadd e \axlabel{cka-seq-comm}\\
        \un \tadd e &\equiv e \axlabel{cka-seq-unit}\\
        e \tjoin (f \tjoin g) &\equiv (e \tjoin f) \tjoin g \axlabel{cka-plus-ass}\\
        e \tjoin f &\equiv f \tjoin e \axlabel{cka-plus-comm}\\
        e \tjoin e &\equiv e \axlabel{cka-plus-idem}
      \end{align}
    \end{minipage}%
    \begin{minipage}[t]{.5\linewidth}
      \begin{align}
        \zero \tjoin e &\equiv e \axlabel{cka-plus-unit}\\
        \zero \tadd e &\equiv\zero \axlabel{cka-seq-absorbing}\\
        e \tadd (f \tjoin g) &\equiv (e \tadd f) \tjoin (e \tadd g) \axlabel{cka-left-distr}\\
        \un\tjoin e\tadd \titerate e&\leqq \titerate e\axlabel{cka-star}\\
        e\tadd f\leqq f&\Rightarrow \titerate e\tadd f \leqq f\axlabel{cka-lfp}
      \end{align}
    \end{minipage}
  \end{tableequations}
  \caption{Axioms of commutative Kleene algebra}
  \label{tab:axioms}
\end{table}
We consider the \emph{axiomatic equivalence relation}
$\equiv$, defined as the smallest congruence on expressions
containing the axioms listed in Table~\ref{tab:axioms}. If $e\equiv f$
we say that $e$ is provably equal to $f$. We use the convention that
$e\leqq f$ means $e\tjoin f\equiv f $.
It is a simple exercise to check that each of these axioms is sound,
meaning that we have:
\begin{equation}
  \forall e,f\in\reg,\,e\equiv f\Rightarrow\sem e = \sem f.\label{eq:soundness}
\end{equation}
As we will prove in Section~\ref{sec:finite-complete}, for finitary expressions, i.e. expressions that do not use the operator $\titerate\_$, we also have completeness:
\begin{equation}
  \forall e,f\in\reg^{fin},\,e\equiv f\Leftrightarrow\sem e = \sem f.\label{eq:fincomplete}
\end{equation}
For this reason, we may (and will) dispense with the $\expr\_$
notation, and identify the vector $\vec u$ and the expression
$\expr{\vec u}$. We also identify a finite set $E$ of expressions (or
a finite set of vectors) with the expression $\tJoin_{e\in E} e$. This
does not introduce ambiguity, thanks to the properties of $\tjoin$, in
particular associativity~\eqref{cka-plus-ass},
commutativity~\eqref{cka-plus-comm}, and
idempotency~\eqref{cka-plus-idem}.

A \emph{linear} expression is a term $e\in\reg$ of the form
$e=\vec u\tadd\titerate{\Base}$, for some vector
$\vec u\in\Vectors$ and finite set $\Base\subseteq\Vectors$. A
\emph{semilinear} expression is a finite \union of linear expressions.
A linear expression $\vec u\tadd\titerate{\Base}$ is said to be
\emph{unambiguous} when $\Base$ is independent. 

The \emph{dimension} of a linear expression
$\vec u\tadd\titerate{ \Base}$ is the cardinal of $\Base$. The
dimension of a semilinear expression $e$, written $\dim e$, is the
maximum of the dimensions of the linear expressions composing it:
$$\dim{\tJoin_{i\in I}\vec u_i\tadd\titerate{ \Base_i}}\eqdef\max\setcompr{\Card{\Base_i}}{i\in I}.$$
\subsection{Overview of the completeness proof}
\label{sec:plan}

In this section we provide an overview of the proof that two
expressions that share the same semantics are provably equal. Since
$\leqq$ is antisymmetric with respect to $\equiv$, it is enough to
show that if the semantics of $e$ is contained in that of $f$, then
the inequality $e\leqq f$ is derivable from
\eqref{cka-seq-ass}-\eqref{cka-lfp}.

The first step of the proof is the following proposition:
\begin{restatable*}{proposition}{decomposition}\label{lemma:decomposition}
  Every \commutative regular expression is provably equal to a finite \union of unambiguous linear expressions.
\end{restatable*}
\begin{proof}[Sketch]
  To prove this result, we will need two steps, first splitting
  expressions into finite \unions of linear expressions (i.e. semilinear
  expressions), and then splitting single linear expressions into finite
  \unions of unambiguous expressions. This later step more technically
  involved, and relies on an induction on the dimension of semilinear
  terms.\qed
\end{proof}
\begin{remark}
  This entails that the sets of vectors generated by \commutative
  regular expressions are exactly the semilinear sets.
  Another consequence is that any semilinear set can be built as a
  finite union of linear sets generated by \emph{independent}
  families.
\end{remark}

We then discharge the case where $e$ and $f$ are both linear
(unambiguity does not play a role here). This proof is fairly
straightforward.
\begin{restatable*}{proposition}{baseInclusion}\label{lemma:baseInclusion}
  Given two linear expressions $e,f$, if $\sem e\subseteq \sem f$ then $e\leqq f$.
\end{restatable*}
Using Proposition~\ref{lemma:decomposition}, together with the fact
that~$\tjoin$ is a join operator with respect to~$\leqq$, we may
extend this seamlessly to the containment of an arbitrary expression
inside a single linear expression.
\begin{restatable*}{corollary}{corBaseInclusion}\label{cor:baseInclusion}
  For any terms $e,f$ such that $f$ is linear and
  $\sem e\subseteq \sem f$, then $e\leqq f$.
\end{restatable*}

We then arrive to the most subtle part of the proof:
\begin{restatable*}{proposition}{baseSplit}\label{lemma:baseSplit}
  Let $f$ be an unambiguous linear expression. For any expression $e$
  there are expressions $\brack{e\wedge f},\brack{e\setminus f}$ such that
  \begin{deflist}
  \item $e\equiv \brack{e\wedge f}\tjoin \brack{e\setminus f}$,
  \item $\sem{\brack{e\wedge f}}\subseteq\sem f$, and
  \item $\sem{\brack{e\setminus f}}\subseteq\sem e\setminus\sem f$.
  \end{deflist}
\end{restatable*}
\begin{proof}[Sketch]
  To prove this proposition, we first discharge the case where
  $e=\vec u\tadd\titerate\Aase$ and $f=\vec v\tadd\titerate\Base$ are
  such that $\Aase\subseteq\iterate \Base$. We call this situation
  ``$e$ is compatible with $f$'', and prove it by induction on the
cardinality of $\Aase$.

Given a fixed unambiguous linear expression
$f=\vec v\tadd\titerate\Base$, we then show that any expression $e$
can be split into a \union $e'\tjoin e''$, where $e'$ is a \union of
compatible expressions, and
$\sem {e''}\subseteq\sem e\setminus\sem f$.  For this task, we
consider $\Vectors$ as a subset of $\Vectors(\Rat)$, and use this
point of view to extend the independent family $\Base$ into a basis
$\BBase$ of the space $\Vectors(\Rat)$. This allows to to have a
bijection between $\Vectors(\Rat)$ and the rational $\BBase$-points,
meaning in particular that every vector $\vec u\in\Vectors$ has unique
``coordinates'' with respect to the vectors $\vec b\in\BBase$. We may
now obtain a characterisation of $\iterate\Base$ in terms of these
coordinates: $\vec u\in\iterate \Base$ iff $\vec u$ has positive
integer coordinates for each $\vec b\in\Base$, and $0$ coordinates for
each $\vec b\in\BBase\setminus \Base$. Thanks to the properties of
$\Nat$ inside $\Rat$, we know that for any vector $\vec u\in\Base$,
there is a number $n\in\Nat$ such that $\scale n{\vec u}$ has integer
coordinates for each $\vec b\in\BBase$. Therefore the crux of the
argument revolves around the \emph{sign} of the coordinates.

The most challenging lemma of this development tackles this very
question. It states that every expression may be rewritten as a \union
$\tJoin_i\vec u_i\tadd\titerate{\Aase_i}$ where each $\Aase_i$ is
\emph{homogeneous}. A family of vectors $\Aase$ is called homogeneous
if for each $\vec b\in\BBase$, either every vector in $\Aase$ has
uniformly positive $\vec b$-coordinates, or uniformly negative ones.
In the first case of positive $\Base$-coordinates and null
$\paren{\BBase\setminus\Base}$-coordinates, the expression may be
massaged into a compatible form. In the other cases, we may rewrite
the expression into the \union of an expression contained in
$\sem e\setminus\sem f$, and one of strictly smaller dimension. We may
therefore conclude the proof of Proposition~\ref{lemma:baseSplit} by
an induction on the dimension.\qed
\end{proof}
Using these results, we may conclude our development:
\begin{theorem}[Completeness of commutative Kleene algebra]
  
  \noindent%
  The axioms \eqref{cka-seq-ass}-\eqref{cka-lfp} are sound and
  complete for the equational theory of semilinear sets.
\end{theorem}
\begin{proof}
  Since soundness is straightforward, we will only focus on
  completeness. As we noticed earlier, we may restrict our attention
  to \emph{inclusion} rather that \emph{equivalence}, relying on
  antisymmetry to conclude.

  Let $e,f\in\reg$ be two expressions such that the semantics of $e$
  is contained in that of $f$. Using
  Proposition~\ref{lemma:decomposition} we may rewrite $f$ as a finite
  \union of unambiguous expressions $\tJoin_{1\leqslant i\leqslant n} f_i$.  We then leverage
  Proposition~\ref{lemma:baseSplit}, to decompose $e$ in terms of the
  $f_i$ as follows:
  \begin{mathpar}
    g_0\eqdef e\and
    g_{i+1}\eqdef \brack{g_i\setminus f_{i+1}}\and
    e_{i+1}\eqdef \brack{g_i\wedge f_{i+1}}.
  \end{mathpar}
  By construction, observe that we have
  \begin{mathpar}
    g_i\equiv e_{i+1}\tjoin g_{i+1}\and
    \sem{e_{i+1}}\subseteq \sem{f_{i+1}}\and
    \sem{g_{i+1}}\subseteq \sem {g_i}\setminus\sem{f_{i+1}}.
  \end{mathpar}
  Therefore we obtain that $e\equiv g_0\equiv e_1\tjoin g_1\equiv\cdots\equiv e_1\tjoin\cdots\tjoin e_n\tjoin g_n$ and:
  \begin{equation*}
    \sem{g_n}\subseteq\sem e\setminus\sem {f_1}\setminus\sem{f_2}\dots\setminus\sem{f_n}=\sem e\setminus \sem f.
  \end{equation*}
  Since we assumed $\sem e\subseteq\sem f$, we know that
  $\sem {g_n}\subseteq\sem e\setminus \sem f=\emptyset$.  We may prove
  by induction on $g_n$ that $g_n\equiv \zero$ (a proof is provided in
  Appendix~\ref{sec:proofs:2.3}).
  
  To conclude, we use Proposition~\ref{lemma:baseInclusion} to show
  that for each $i$, $e_i\leqq f_i$, thus showing that
  \begin{equation*}
    e\equiv \tJoin_i e_i\tjoin g_n\leqq \tJoin_i f_i\tjoin\zero\equiv f.\tag*{\qed}
  \end{equation*}

\end{proof}
\section{Preliminary results}
\label{sec:prelim}

\subsection{Completeness in the finite case}
\label{sec:finite-complete}
In this section we prove the obvious. The point is to make explicit
the techniques and steps that are necessary, or at least useful, to
establish statements that are instrumental for the main proof of this
paper.

\begin{lemma}\label{lem:add-vect}
  $\forall\vec u,\vec v\in\Vectors,\,\expr{\vec u\vecplus\vec v}\equiv\expr{\vec u}\tadd\expr{\vec v}$.
\end{lemma}
\begin{proof}
  By induction on $\alphabet$:
  \begin{description}
  \item[$\blacktriangleright \alphabet=\set{a}$:]
    in this case, we only need to use associativity of $\tadd$, i.e. axiom~\eqref{cka-seq-ass}, to prove that $\expr{\scale {\paren{n+m}} a}\equiv\expr{\scale n a}\tadd\expr{\scale m a}$.
  \item[$\blacktriangleright \alphabet=\set{a}\uplus\alphabet'$:]
    in this case, we have:
    \begin{mathpar}
      \expr{\vec u}=\scale {\vec u_a}a\tadd\expr{\vec u'}\and
      \expr{\vec v}=\scale {\vec v_a}a\tadd\expr{\vec v'}\and
      \expr{\vec {u\vecplus v}}=\scale {\paren{\vec u_a+\vec v_a}}a\tadd\expr{\vec u'\vecplus\vec v'}
    \end{mathpar}
    where $\vec u',\vec v'\in\Vectors[\alphabet']$.
    By induction we get $\expr{\vec u'\vecplus\vec v'}\equiv\expr{\vec u'}\tadd\expr{\vec v'}$.
    \begin{align*}
      \expr{\vec u\vecplus\vec v}
      =\scale {\paren{\vec u_a+\vec v_a}}a\tadd\expr{\vec u'\vecplus\vec v'}
      &\equiv\scale {\paren{\vec u_a+\vec v_a}}a\tadd\expr{\vec u'}\tadd\expr{\vec v'}\tag{by~I.H.}\\
      &\equiv\scale {\vec u_a}a\tadd\expr{\vec u'}\tadd\scale{\vec v_a}a\tadd\expr{\vec v'}
        \tag{by~\ref{cka-seq-ass},\ref{cka-seq-comm}}\\
      &=\expr{\vec u}\tadd\expr{\vec v}.\tag*{\qed}
    \end{align*}
  \end{description}
\end{proof}

\begin{restatable}{lemma}{splitfinite}\label{lem:split-finite}
  For every expression in $\reg^{fin}$, it holds that
  $e\equiv \tJoin_{\vec v\in\sem e}\expr{\vec v}$.
\end{restatable}
The proof of this lemma proceeds by a straightforward induction on
expressions. We make this explicit in Appendix~\ref{sec:proofs:3.1}

\begin{lemma}\label{lem:ax-mem}
  For any vector $\vec u\in\Vectors$ and any expression $e\in\reg$, we have:
  $$\vec u\in\sem e\Rightarrow\expr {\vec u}\leqq e.$$
\end{lemma}
\begin{proof}
  By induction on $e$:
  \begin{description}
  \item[$\blacktriangleright \zero,\,\un,\,a$:] these cases hold trivially.
  \item[$\blacktriangleright f\tjoin g$:] $\vec u\in\sem{f\tjoin g}=\sem f\join\sem g$ implies that either $\vec u\in\sem f$, in which case we have $\expr{\vec u}\leqq f\leqq f\tjoin g$, or $\vec u\in\sem g$, in which case we have $\expr{\vec u}\leqq g\leqq f\tjoin g$
  \item[$\blacktriangleright f\tadd g$:] $\vec u\in\sem{f\tadd g}=\sem f\vecplus\sem g$ implies that there are vectors $\vec v,\vec w$ such that
    $\vec u =\vec v\vecplus\vec w$, $\vec v \in\sem f$, and $\vec w\in\sem g$.
    We conclude this case:
    \begin{align*}
      \expr{\vec u}
      =\expr{\vec v\vecplus\vec w}
      &\equiv\expr{\vec v}\tadd\expr{\vec w}\tag{By Lemma~\ref{lem:add-vect}}\\
      &\leqq f\tadd g\tag{By I.H.}
    \end{align*}
  \item[$\blacktriangleright \titerate f$:] since,
    $\vec u\in\sem{\titerate f}=\iterate{\sem f}$, $\vec u$ may be
    decomposed as $\vec u =\vec u_1\vecplus\cdots\vecplus\vec u_n$,
    with $\forall i,\vec u_i \in\sem f$. We conclude this proof:
    \begin{align*}
      \expr{\vec u}
      =\expr{\vec u_1\vecplus\cdots\vecplus\vec u_n}
      &\equiv\expr{\vec u_1}\tadd\cdots\tadd\expr{\vec u_n}
        \tag{By Lemma~\ref{lem:add-vect}}\\
      &\leqq f\tadd\cdots\tadd f \tag{By I.H.}\\
      &\leqq \titerate f\tag*{\qed}
    \end{align*}
  \end{description}
\end{proof}

\begin{corollary}\label{cor:fin-incl}
  $\forall e\in\reg^{fin},\forall f\in\reg,\,\sem e\subseteq\sem f\Rightarrow e\leqq f.$
\end{corollary}
\begin{proof}
  Thanks to Lemma~\ref{lem:split-finite}, we know that:
  $e\equiv \tJoin_{\vec v\in\sem e}\expr{\vec v}$.
  Since $\sem e\subseteq\sem f$, by Lemma~\ref{lem:ax-mem} we know that for all $\vec v\in\sem e$, we can derive $\expr v\leqq f$. Therefore, we get the following proof:
  \begin{equation*}
    e\equiv \tJoin_{\vec v\in\sem e}\expr{\vec v}\leqq\tJoin_{\vec v\in\sem e} f\equiv f
    \tag*{\qed}
  \end{equation*}
\end{proof}
\begin{corollary}\label{cor:completeness-fin}
  $\forall e,f\in\reg^{fin},\,e\equiv f\Leftrightarrow\sem e = \sem f.$
\end{corollary}
\begin{proof}
  The left to right implication is soundness, which we already stated
  to hold for any expressions. For the converse direction, notice that
  $e\equiv f\Leftrightarrow e\leqq f \wedge f\leqq e$, so we obtain
  the desired entailment by two applications of
  Corollary~\ref{cor:fin-incl}.  \qed
\end{proof}

\subsection{Laws of commutative Kleene algebra}
\label{sec:stuff}
Omitted proofs from this section are provided in Appendix~\ref{sec:proofs:3.2}.

The following are laws of Kleene algebra. Since a
commutative Kleene algebra is in particular a Kleene algebra, these
hold here as well.
\begin{align}
  \titerate{e}\tadd\titerate e&\equiv \titerate e\equiv \ptiterate{\titerate e}\law{star}
  \\
  \ptiterate{e\tadd \titerate f}&\equiv\un\tjoin \jparen{e\tadd\ptiterate{e\tjoin f}}\law{iter-sum}\\
\titerate e&\equiv e^{<n}\tjoin \jparen{\scale n e\tadd \titerate e}\law{n-or-more}
\end{align}
Here the notation $e^{<n}$ refers to the expression
$e^{<n}\eqdef\scale{\paren{n-1}}{\paren{e\tjoin \un}}$, with the
convention that $e^{<0}=\zero$. Clearly if $n>0$ we have
$\un\leqq e^{<n}\leqq \titerate e$. Note also that for any $k\in\Nat$
we have $e^{<k+1}\equiv \scale k e\tjoin e^{<k}$.

We show the proof of the following principle of Kleene algebra. This
was used as an infinitary axiom scheme in both Redko and Pilling's
proofs.
\begin{restatable}{lemma}{starcontinuity}\label{lem:star-cont}
  For any expression $e\in\reg$ and any positive number $n>0$ the following holds:
  \begin{equation}
    \titerate{e}\equiv e^{< n}\tadd\ptiterate{\scale n e}\law{first-iterates}.
  \end{equation}
\end{restatable}
\begin{proof}
  We prove both inequalities, relying on antisymmetry of $\leqq$ to conclude.
  \begin{description}
  \item[($\geqq$):] clearly if $n>0$, we have: $\scale n e\leqq \scale n {\titerate e}\equiv \titerate e$. Also notice that $$e^{<n}=\scale{\paren{n-1}}{\paren{e\tjoin \un}}\leqq \scale{\paren{n-1}}{\titerate{e}}\leqq \titerate e$$
    Therefore we have the following:
    $e^{< n}\tadd\ptiterate{\scale n e}\leqq \titerate e\tadd \ptiterate{\titerate e}\equiv \titerate e$.
  \item[($\leqq$):]
    since $n>0$, we have $\un\leqq e^{<n}\equiv\scale{\paren{n-1}}e\tjoin e^{<n-1}$, hence:
    \begin{align*}
      e\tadd e^{<n}\tadd\ptiterate{\scale n e}
      &\equiv e\tadd\paren{\scale{\paren{n-1}}e\tjoin e^{<n-1}}\tadd \ptiterate{\scale n e}\\
      &\equiv \jparen{e\tadd\scale{\paren{n-1}}e\tadd \ptiterate{\scale n e}}\tjoin \jparen{e\tadd e^{<n-1}\tadd \ptiterate{\scale n e}}\\
      &\equiv \jparen{\scale n e\tadd \ptiterate{\scale n e}}\tjoin \jparen{e\tadd e^{<n-1}\tadd \ptiterate{\scale n e}}\\
      &\leqq \ptiterate{\scale n e}\tjoin \jparen{\paren{e\tjoin \un}\tadd e^{<n-1}\tadd \ptiterate{\scale n e}}\\
      &\equiv\jparen{\un\tadd\ptiterate{\scale n e}}\tjoin \jparen{e^{<n}\tadd \ptiterate{\scale n e}}\\
      &\leqq e^{<n}\tadd \ptiterate{\scale n e}.
    \end{align*}
    By~\eqref{cka-lfp}, this entails
    $\titerate e \tadd e^{<n}\tadd\ptiterate{\scale n e}\leqq
    e^{<n}\tadd\ptiterate{\scale n e}$, so we can now conclude since:
    $\titerate e \equiv\titerate e \tadd\un\leqq\titerate e \tadd
    e^{<n}\tadd\ptiterate{\scale n e}\leqq
    e^{<n}\tadd\ptiterate{\scale n e}$.\qed
  \end{description}
\end{proof}

We will also use the following law of commutative Kleene algebra.
\begin{align}
  \ptiterate{e\tjoin f}&\equiv\titerate e\tadd\titerate f\law{sum-to-prod}
\end{align}

\begin{restatable}{lemma}{basevector}\label{lem:base-vector}
  Given a finite set $\Base=\set{\vec u_1,\dots,\vec u_n}$ and a point
  $p\in\Points\setminus\unitvec$, the following holds:
  $$\titerate{\Base}\equiv \titerate{\tovec p}\tadd\tJoin_{i}\paren{\vec u_i^{<p_i}\tadd\tsum_{i\neq j}\titerate{\vec u_j}}.$$
\end{restatable}
Due to the length of this proof, we are unable to reproduce it
here. The interested reader may find on extended versions of this
abstract. As a corollary, we get the following statement:
\begin{corollary}\label{cor:base-vector}
  Given a finite set of vectors $\Base$ and a vector
  $\vec w\in\iterate\Base$, there exists a semilinear expression $e$
  such that
  \begin{deflist}
  \item $\titerate{\Base}\equiv \titerate{\vec w}\tadd e$, and
  \item $\dim e<\Card\Base$.
  \end{deflist}
\end{corollary}

\begin{remark}[On Pilling's axiomatisation]
  The axiomatisation proved complete by Pilling differs from ours in
  several ways. It does not include~\eqref{cka-lfp}
  or~\eqref{cka-star}, but is instead entirely composed of identity
  axioms. It is however infinite, including the identities from
  Lemma~\ref{lem:star-cont} for each value of $n$. Besides those, it
  includes our axioms \eqref{cka-seq-ass}-\eqref{cka-left-distr},
  together with \eqref{eq:sum-to-prod} and the following laws, all of
  which are derivable from \eqref{cka-seq-ass}-\eqref{cka-lfp}:
  \begin{mathpar}
    \titerate\un\equiv\un\and \ptiterate{e\tadd\titerate
      f}\equiv\un\tjoin e\tadd\titerate e\tadd\titerate f\and
    \ptiterate{e\tjoin f}\equiv\ptiterate{e\tadd
      f}\tadd\ptiterate{\titerate e\tjoin\titerate f}.
  \end{mathpar}
\end{remark}

\subsection{Rational vector spaces}
\label{sec:rat}
We will use in our development two facts about $\Rat$-vector spaces\footnotemark.
\begin{restatable}{remark}{linearindep}\label{lem:linear-indep}
  Any finite set of Parikh vectors $\Base\subseteq\Vectors$ is
  independent according to our definition if and only if it is
  linearly independent inside the space $\Vectors(\Rat)$.
\end{restatable}\footnotetext{We provide a detailed proof of
  Remark~\ref{lem:linear-indep} in Appendix~\ref{sec:proofs:3.3}.}

\begin{remark}\label{lem:extend-basis}
  Let $\Base\subseteq\Vectors$ be a finite independent set of vectors. There
  exists another finite set $\Base'\subseteq\Vectors$ such that
  \begin{deflist}
  \item $\Base$ and $\Base'$ are disjoint,
  \item $\Base\join\Base'$ is independent, and
  \item $\Base\join\Base'$ is a basis of $\Vectors(\Rat)$.
  \end{deflist}
\end{remark}
\begin{proof}
  This is an instance of the \emph{incomplete basis theorem}:
  
  Let $E$ be a vector space, $G$ a spanning family of $E$ and $L$ a linearly independent set. Then there exists $F \subset G\setminus L$ such that $L\cup F$ is a basis of $E$.

  In our case, we may choose the family $G$ to be the canonical basis
  on $\Vectors(\Rat)$, i.e.
  $\setcompr{\tovec a}{a\in\alphabet}\subseteq\Vectors$.\qed
\end{proof}
\section{Decomposition into linear expressions}
\label{sec:decomp}

In this section, we prove the first statement in the proof,
i.e. Proposition~\ref{lemma:decomposition}, that states that every
expression is provably equal to as a finite \union of unambiguous
expressions.

First, we split expressions into finite \unions of linear expressions,
i.e. semilinear expressions. This first step already entails that
every commutative regular language has star-height at most one.

\begin{lemma}\label{lem:decomp}
  Any expression $e$ is provably equal to some semilinear expression.
\end{lemma}
\begin{proof}
  We will show this by induction on expressions.
  \begin{description}
  \item[$\blacktriangleright \zero,\un, a ,e\tjoin f$:] these all hold
    trivially with $\zero$ as an empty \union,
    $\un\equiv\un\tadd\titerate\zero$ and
    $a\equiv a\tadd\titerate\zero$. For $e\tjoin f$, since a
    semilinear expression is defined as a \union, we can simply take
    the \union of the semilinear expressions computed inductively for
    $e$ and $f$.
  \item[$\blacktriangleright e\tadd f$:] we simply make the following computation:
    \begin{align*}
      e\tadd f
      &\equiv\tJoin_{1\leqslant i\leqslant n}\jparen{\vec u_i\tadd\titerate{\Base_i}}
        \tadd
        \tJoin_{1\leqslant i\leqslant m}\jparen{\vec u'_i\tadd\titerate{\Base'_i}}\\
      &\equiv\tJoin_{1\leqslant i\leqslant n,1\leqslant j\leqslant m}\jparen{\vec u_i\tadd\titerate{\Base_i}\tadd\vec u'_i\tadd\titerate{\Base'_i}}\\
      &\equiv\tJoin_{1\leqslant i\leqslant n,1\leqslant j\leqslant m}\jparen{\vec u_i\tadd \vec v_j\tadd\titerate{\paren{\Base_i\join\Base'_j}}}.\tag{by~\ref{eq:sum-to-prod}}
    \end{align*}
  \item[$\blacktriangleright \titerate e$:] we perform a similar
    computation, relying on~\eqref{eq:sum-to-prod}
    and~\eqref{eq:iter-sum}:
    \begin{align*}
      \titerate e\equiv\ptiterate{\tJoin_{1\leqslant i\leqslant n}\jparen{\vec u_i\tadd\titerate{\Base_i}}}
      &\equiv\tsum_{1\leqslant i\leqslant n}\ptiterate{\vec u_i\tadd\titerate{\Base_i}}
        \tag{by~\ref{eq:sum-to-prod}}\\
      &\equiv\tsum_{1\leqslant i\leqslant n}\paren{\un\tjoin\jparen{\vec u_i\tadd\ptiterate{\vec u_i\tjoin\Base_i}}}\tag{by~\ref{eq:iter-sum}}\\
      &\equiv\tJoin_{I\subseteq\set{1,\dots,n}}\paren{\tsum_{i\in I}\vec u_i}\tadd\ptiterate{\tJoin_{i\in I}\paren{\vec u_i\tjoin\Base_i}}\tag*{\qed}
    \end{align*}
  \end{description}
\end{proof}

The second step consists in splitting individual linear expressions
into \unions of unambiguous ones. We will do the decomposition by
induction on the dimension of the expressions. In an attempt to
clarify the argument, we prove separately the following lemma.

\begin{lemma}\label{lem:dep-inf-dim}
  Let $\Base$ be a finite set of vectors. If $\Base$ is not independent, there exists a semilinear expression $e$ such that
  \begin{deflist}
  \item $\titerate{\Base}\equiv e$, and
  \item $\dim e<\Card\Base$.
  \end{deflist}
\end{lemma}
\begin{proof}
  If $\Base$ is not independent, there is a pair of $\Base$-points
  $\alpha,\beta\in\Points$ such that $\alpha\neq\beta$ and
  $\tovec \alpha=\tovec \beta$. We define the $\Base$-points $\gamma$,
  $\mu$, and $\nu$ by
  \begin{mathpar}
    \gamma_i\eqdef\min\paren{\alpha_i,\beta_i},\and
    \mu_i\eqdef\alpha_i-\gamma_i,\and
    \nu_i\eqdef\beta_i-\gamma_i.
  \end{mathpar}
  Observe that by construction $\mu$ and $\nu$ denote the same
  $\Base$-point, and that for every coordinate either $\mu_i$ or
  $\nu_i$ equals $0$. We split $\Base$ according to $\mu$ and $\nu$:
  \begin{mathpar}
    \Base_\mu\eqdef\setcompr{\vec u\in\Base}{\mu_{\vec u}>0}\and
    \Base_\nu\eqdef\setcompr{\vec u\in\Base}{\nu_{\vec u}>0}\and
    \Base_0\eqdef\setcompr{\vec u\in\Base}{\mu_{\vec u}=\nu_{\vec u}=0}.
  \end{mathpar}
  Notice that since $\Base = \Base_\mu\uplus\Base_\nu\uplus\Base_0$,
  and thanks to~\eqref{eq:sum-to-prod}, we have:
  $$\titerate{ \Base}\equiv\ptiterate{\Base_\mu\tjoin\Base_\nu\tjoin\Base_0}\equiv\titerate{\Base_\mu}\tadd\titerate{\Base_\nu}\tadd\titerate{\Base_0}$$
  Let $\vec w\eqdef\tovec\mu=\tovec \nu$. Since this word is both in
  $\iterate{\Base_\mu}$ and in $\iterate{\Base_\nu}$, by
  Corollary~\ref{cor:base-vector} there are semilinear expressions
  $e_\mu$ and $e_\nu$ such that:
  \begin{mathpar}
    \titerate{\Base_\mu}\equiv \titerate{\vec w}\tadd e_\mu\and
    \titerate{\Base_\nu}\equiv \titerate{\vec w}\tadd e_\nu\and
    \dim{e_\mu}<\Card{\Base_\mu}\and
    \dim{e_\nu}<\Card{\Base_\nu}.
  \end{mathpar}
  Combining these facts together we obtain:
  $$\titerate{ \Base}\equiv \titerate{\vec w}\tadd e_\mu\tadd \titerate{\vec w}\tadd e_\nu\tadd\titerate{\Base_0}\equiv \titerate{\vec w}\tadd e_\mu\tadd e_\nu\tadd\titerate{\Base_0}.$$
  By distributivity, this expression may be rewritten as a semilinear
  expression $e$, the dimension of which is the sum of the dimensions
  of its terms, i.e.:
  \begin{align*}
    \dim e&=\dim{\titerate{\vec w}}+\dim{e_\mu}+\dim{e_\nu}+\dim{\titerate{\Base_0}}\\
          &=1+\dim{e_\mu}+\dim{e_\nu}+\Card{\Base_0}\\
          &\leqslant 1+\paren{\Card{\Base_\mu}-1}+\paren{\Card{\Base_\nu}-1}+\Card{\Base_0}\\
          &= \paren{\Card{\Base_\mu}+\Card{\Base_\nu}+\Card{\Base_0}}-1=\Card\Base-1<\Card\Base.\tag*{\qed}
  \end{align*}
\end{proof}

We may now use this lemma in an induction on the dimension of
semilinear expressions to obtain the desired result.

\decomposition
\begin{proof}
  Let $e$ be a commutative regular expression. By
  Lemma~\ref{lem:decomp}, we can compute a semilinear expression $f$
  such that $e\equiv f$. We prove by induction on $\dim f$ that $f$
  can be written as a \union of unambiguous linear expressions. If $f$
  is already a \union of unambiguous linear expressions, then the
  statement holds. Otherwise, let $\vec u\cdot \titerate{ \Base}$ be a
  term in $f$ such that $\Base$ is not independent. Notice that by
  definition of the dimension of an expression, we have
  $\Card \Base\leqslant\dim f$. Thanks to Lemma~\ref{lem:dep-inf-dim},
  we can obtain a semilinear expression~$f'$ such that:
  $\titerate{ \Base}\equiv f'$ and $\dim{f'}<\Card\Base$. By
  induction, $f'$ can be rewritten as a \union of unambiguous linear
  expressions, and by distributivity so can $\vec u\tadd f'$. We
  repeat this argument for every term in $f$, and take the \union of
  the resulting decompositions to obtain a \union of unambiguous linear
  expressions that is provably equal to $f$, and so to $e$.\qed
\end{proof}

\section{Inclusion of linear terms}
\label{sec:incl-bases}
This step is the easiest in this development. We prove it directly.
\baseInclusion
\begin{proof}
  Let $e=\vec u\tadd\titerate{\Aase}$ and
  $f=\vec v\tadd\titerate{\Base}$. Recall that the pointwise
  ordering of vectors of natural numbers is a well-quasi ordering,
  meaning in particular that every infinite set contains at least one
  ordered pair.

  Let $\vec a\in\Aase$. We now show that
  $\exists k_{\vec a}\geqslant 1:\,\scale k{\vec
    a}\in\iterate\Base$. Consider the set
  $V_{\vec a}\eqdef\setcompr{p\in\Points}{\vec v\vecplus\tovec p\in\vec
    u\vecplus\iterate{\vec a}}$. Since
  $\vec u\vecplus\iterate{\vec a}\subseteq\sem e\subseteq\sem f$, we know
  that this set is infinite. Therefore, there are two points
  $p,q\in V_a$ that are pointwise ordered, which implies that there
  exists a third point $r\in\Points\setminus\unitvec$ such that
  $p\vecplus r=q$. Since $p,q\in V_a$, there are numbers $n,m\in\Nat$ such
  that $\vec v\vecplus\tovec p=\vec u\vecplus\scale n{\vec a}$ and
  $\vec v\vecplus\tovec q=\vec u\vecplus\scale m{\vec a}$. Since $p$ is
  pointwise smaller than $q$, it follows that $n<m$, i.e.
  $1\leqslant m-n$. Coincidentally, since $p\vecplus r=q$, we get that
  $\tovec r=\tovec q-\tovec p$, i.e.
  $\tovec r=\scale {\paren{m-n}}{\vec a}$. Therefore,
  $k_{\vec a}\eqdef m-n$ satisfies the required properties, namely
  $k_{\vec a}\geqslant 1$ and
  $\scale {k_{\vec a}}{\vec a}=\tovec r\in\iterate\Base$.

  This is enough to complete the proof:
  \begin{align*}
    e=\vec u\tadd\titerate{\Aase}
      &\equiv \vec u\tadd\tsum_{\vec a\in\Aase}\paren{\vec a^{<k_{\vec a}}\tadd\ptiterate{\scale {k_{\vec a}}{\vec a}}}
        \tag{by~\ref{eq:sum-to-prod},\ref{eq:first-iterates}}\\
      &\leqq \vec u\tadd\tsum_{\vec a\in\Aase}\paren{\vec a^{<k_{\vec a}}\tadd\ptiterate{\titerate{ \Base}}}
        \tag{by Lemma~\ref{lem:ax-mem}}\\
      &\equiv \paren{\vec u\tadd\tsum_{\vec a\in\Aase}\vec a^{<k_{\vec a}}}\tadd\titerate{ \Base}\\
      &\leqq \vec v\tadd\titerate{ \Base}\tadd\titerate{ \Base}
        \tag{by Corollary~\ref{cor:fin-incl}}
        \equiv \vec v\tadd\titerate{ \Base}=f.
  \end{align*}\qed
\end{proof}

\corBaseInclusion
\begin{proof}
  By Proposition~\ref{lemma:decomposition} we can write $e\equiv E$, with
  $E$ a finite set of (unambiguous) linear expressions.  To obtain
  $e\leqq f$, we only need to show for each $g\in E$ that $g\leqq f$.
  Since $g\leqq E\equiv e$, by soundness we have
  $\sem g\subseteq \sem e$.  Therefore we have $g,f$ linear and
  $\sem g\subseteq\sem e\subseteq \sem f$: by
  Proposition~\ref{lemma:baseInclusion} we get $g\leqq f$.\qed
\end{proof}
\section{Intersection and difference}
\label{sec:inter-diff}

We fix for the remainder of this section an unambiguous linear term
$f\eqdef \vec v\tadd\titerate{\Base}$. A \emph{decomposition} of an
expression $e$ is a pair of terms $\tuple{x,y}$ such that
\begin{deflist}
\item $e\equiv x\tjoin y$, 
\item $\sem x\subseteq \sem f$, and
\item $\sem y\subseteq \sem e\setminus \sem f$.
\end{deflist}

\begin{remark}\label{rmk:dec-union}
  It is useful to keep in mind that the operations $\_\cap X$ and
  $\_\setminus X$ commute with unions.
  Because of this, showing that every expression can be decomposed is
  equivalent to proving that every \emph{linear} expression is
  decomposable.
  Indeed, using Lemma~\ref{lem:base-vector}, we may write any $e$ as a
  finite \union of linear expressions $e_1,\dots,e_n$.
  If we have decompositions $\tuple{x_i,y_i}$ of each of those terms,
  then the pair $\tuple{\tJoin_i x_i,\tJoin_i y_i}$ is a decomposition
  of $e$.
\end{remark}

Now we show that linear expressions that are in some sense
``compatible'' with $f$ can be decomposed.
\begin{lemma}\label{lem:dec-compatible}
  A linear expression $e=\vec u\tadd\titerate{\Aase}$ such that
  $\iterate{\Aase}\subseteq \iterate\Base$ can be decomposed.
\end{lemma}
\begin{proof}
  By induction on $\Card \Aase$. If $\Card \Aase=0$, the statement holds
  trivially, with
  $\tuple{x,y}\in\set{\tuple{e,\zero},\tuple{\zero,e}}$ depending on
  whether $\vec u\in\sem f$. Otherwise, if
  $\sem e\cap \sem f=\emptyset$, then again, the statement holds
  trivially, with $x=\zero$ and $y=e$.

  Therefore we only need to consider the case where $\Card\Aase>0$ and
  we have a vector $\vec w\in\sem e\cap\sem f$.  Since
  $\vec w\in\sem e$, there is a point $\alpha\in\Points[\Aase]$ such
  that $\vec w=\vec u\tadd\tovec[\Aase]\alpha$. We make the following
  transformation on $e$, using \eqref{eq:sum-to-prod} and
  \eqref{eq:n-or-more}:
  $$e=\vec u\tadd\titerate{\Aase}\equiv \vec u\tadd\tsum_{\vec
    a\in\Aase}\titerate{\vec a}\equiv \vec u\tadd\tsum_{\vec
    a\in\Aase}\paren{\vec a^{<\alpha(\vec a)}\tjoin\paren{\scale{\alpha(\vec a)}{\vec
        a}\tadd\titerate{\vec a}}}. $$
  For $A\subseteq\Aase$, we define:
  $U_A\eqdef
  \vec u
  \tadd\tsum_{\vec a \notin A}\vec a^{<\alpha(\vec a)}
  \tadd\tsum_{\vec a\in A}\scale{\alpha(\vec a)}{\vec a}$.
  Notice that $\sem {U_A}$ is finite, and that for $A=\Aase$, we get
  $U_\Aase=\vec u\tadd\tsum_{\vec a\in \Aase}\scale{\alpha(\vec a)}{\vec a}=\vec w$.
  By distributivity, we get from the previous identity:
  $$
  e\equiv \tJoin_{A\subseteq \Aase}\tJoin_{\vec t\in U_A}\vec t\tadd\titerate{ A}
  \equiv \paren{\vec w\tadd\titerate{ \Aase}}\tjoin
  \tJoin_{A\subsetneq \Aase}\tJoin_{\vec t\in U_A}\vec t\tadd\titerate{ A}.
  $$
  For each $A\subsetneq \Aase$, we have
  \begin{deflist}
  \item $\Card A<\Card \Aase$
  \item $\iterate A\subseteq\iterate \Aase\subseteq\iterate\Base$.
  \end{deflist}
  Therefore we may use our induction hypothesis to get for each
  $A\subsetneq\Aase$ and $\vec t\in U_A$ a pair of terms
  $x_{\vec t,A}$ and $y_{\vec t,A}$ such that:
  \begin{deflist}
  \item $\vec t\tadd\titerate{ A}\equiv {x_{\vec t,A}}\tjoin {y_{\vec t,A}}$, 
  \item $\sem {x_{\vec t,A}}\subseteq \sem f$, and
  \item $\sem {y_{\vec t,A}}\subseteq \sem {\vec t\tadd\titerate{ A}}\setminus \sem f$.
  \end{deflist}
  Finally, we conclude by setting
  \begin{mathpar}
    x\eqdef \paren{\vec w\tadd\titerate{ \Aase}}\tjoin
    \tJoin_{A\subsetneq \Aase}\tJoin_{\vec t\in U_A}x_{\vec t,A}\and
    y\eqdef \tJoin_{A\subsetneq \Aase}\tJoin_{\vec t\in U_A}y_{\vec t,A}.
  \end{mathpar}
  Clearly, $e\equiv x\tjoin y$. Since $\vec w\in\sem f$,
  $\iterate\Aase\subseteq\iterate \Base$ and
  $f\equiv f\tadd\titerate{\Base}$, we know that
  $\sem{\vec w\tadd\titerate{ \Aase}}\subseteq\sem
  f\tadd\iterate\Base=\sem f$. Therefore we get
  $\sem x\subseteq\sem f$. Finally, we know that
  \begin{equation*}
    \sem y=\Join_{A\subseteq\Aase}\Join_{\vec t\in U_A}\sem{y_{\vec t,A}}
    \subseteq \Join_{A\subseteq\Aase}\Join_{\vec t\in U_A}\sem {\vec t\tadd\titerate{ A}}\setminus \sem f
    \subseteq \sem e \setminus\sem f.\tag*{\qed}
  \end{equation*}
\end{proof}

Using Remark~\ref{lem:extend-basis}, we extend $\Base$ with
${\bar\Base}\subseteq\Vectors$ such that
$\BBase\eqdef\Base\uplus{\bar\Base}$ is a basis of $\Vectors(\Rat)$.
As such, $\tovec[\BBase]\_$ may be seen as a bijection between the
rational $\BBase$-points $\Rat^\BBase$ and the vector space
$\Vectors(\Rat)$. We write $\topoints\_$ for the inverse bijection.
A linear expression $\vec u\tadd\titerate{\Aase}$ is called
\emph{homogeneous} if:
$$\forall \vec b\in\BBase,\forall \vec x,\vec
y\in\Aase,\,\topointsc{\vec x}{\vec b}>
0\Rightarrow\topointsc{\vec y}{\vec b}\geqslant 0.$$

\begin{lemma}\label{lem:homogeneous}
  Every expression is provably equal to a finite \union of homogeneous
  expressions.
\end{lemma}
\begin{proof}
  This proof works by double induction. Thanks to
  Lemma~\ref{lemma:decomposition}, we may write any expression as a
  finite \union of linear expressions. Therefore, it suffices to show
  that the statement holds for linear expressions. Let
  $e\eqdef\vec u\tadd\titerate{\Aase}$. We introduce two more
  definitions:
  \begin{itemize}
  \item\emph{partial homogeneity}: for a subset $B\subseteq\BBase$, $e$ is
    $B$-homogeneous if:
    $$\forall \vec b\in B,\forall \vec x,\vec
    y\in\Aase,\,\topointsc{\vec x}{\vec b}>
    0\Rightarrow\topointsc{\vec y}{\vec b}\geqslant 0.$$
  \item\emph{$\vec b$-score}: for $\vec b\in\BBase$, the $\vec b$-score of
    $e$ is the number
    $\Card{\setcompr{\vec x\in\Aase}{\topointsc{\vec x}{\vec b}\neq 0}}$.
  \end{itemize}
  
  We now prove by induction on $\Card B$ that
  $\forall B\subseteq \BBase$, any linear expression
  $\vec u\tadd\titerate{\Aase}$ is provably equal to a finite
  \union of $B$-homogeneous expressions.

  \noindent%
  $\blacktriangleright B=\emptyset$\textbf{:} the claim holds trivially, since any linear
  expression is $\emptyset$-homogeneous.

  \noindent%
  $\blacktriangleright \vec b\uplus B$\textbf{:} by induction, any expression is provably
  equal to a finite \union of $B$-homogeneous expressions, so what
  remains to show is the following: any $B$-homogeneous linear
  expression is provably equal to a finite \union of
  $\paren{\vec b\uplus B}$-homogeneous expressions. (Notice that
  being $\paren{\vec b\uplus B}$-homogeneous means being both
  $B$-homogeneous and $\vec b$-homogeneous.) This we prove by
  induction on the $\vec b$-score of the expressions.

  \textbf{$\triangleright$ $\vec b$-score $=0$:} in this case the
  expression is already $\vec b$-homogeneous, and since by assumption
  it was $B$-homogeneous, the statement holds.

  \textbf{$\triangleright$ otherwise:} if
  $\vec u\tadd\titerate{\Aase}$ is already $\vec b$-homogeneous the
  statement already holds.  Otherwise, there are
  $\vec x,\vec y\in\Aase$ such that $\topointsc{\vec x}{\vec b}>0$ and
  $\topointsc{\vec y}{\vec b}<0$. By the properties of $\Rat$, there
  are two natural numbers $n,m>0$ such that
  $$\topointsc{\scale n {\vec x}\vecplus\scale m {\vec y}}{\vec b}
  =\topointsc{\scale n {\vec x}}{\vec b}+\topointsc{\scale m {\vec y}}{\vec b}
  =n\times\topointsc{\vec x}{\vec b}+m\times\topointsc{\vec y}{\vec b}=0.$$
  Since $\scale n {\vec x}\vecplus\scale m {\vec y}\in \iterate\Aase$, we have by Lemma~\ref{lem:base-vector}:
  \begin{mathpar}
    \titerate{\Aase}\equiv \ptiterate{\scale n {\vec x}\vecplus\scale m {\vec y}}\tadd\tJoin_{\vec v\in\Aase}\paren{\vec v^{<p_{\vec v}}\tadd\tsum_{\vec w\neq\vec v}\titerate{\vec w}}\and
    \text{with }p_{\vec v}\eqdef\left\{
      \begin{array}{ll}
        n&\text{ if }\vec v=\vec x\\
        m&\text{ if }\vec v=\vec y\\
        0&\text{ otherwise}\\
      \end{array}\right..
  \end{mathpar}
  We may simplify this expression, since if
  $\vec v\neq\vec x,\vec y$ we have $p_{\vec v}=0$, so:
  \begin{align*}
    \vec v^{<p_{\vec v}}\tadd\tsum_{\vec w\neq\vec v}\titerate{\vec w}
    &=\vec v^{<0}\tadd\tsum_{\vec w\neq\vec v}\titerate{\vec w}
      =\zero\tadd\tsum_{\vec w\neq\vec v}\titerate{\vec w}
      \equiv \zero.\\
    \Rightarrow
    \titerate{\Aase}%
    &\equiv \ptiterate{\scale n {\vec x}\vecplus\scale m {\vec y}}%
      \tadd\paren{\jparen{\vec x^{<n}\tadd\tsum_{\vec w\neq\vec x}\titerate{\vec w}}%
      \tjoin\jparen{\vec y^{<m}\tadd\tsum_{\vec w\neq\vec y}\titerate{\vec w}}}.%
  \end{align*}
  Therefore we split $e=\vec u\tadd\titerate{ \Aase}$ into
  two finite families of linear expressions:
  \begin{align*}
    e
    &\equiv
      \paren{\tJoin_{\vec v\in\vec u\tadd\vec x^{<n}}\vec v\tadd\titerate{\Aase_1}}
      \tjoin
      \paren{\tJoin_{\vec v\in\vec u\tadd\vec y^{<m}}\vec v\tadd\titerate{\Aase_1}}\\
    \text{where }
    \Aase_1&\eqdef\set{\scale n {\vec x}\vecplus\scale m {\vec y}}
             \cup\paren{\Aase\setminus\set{\vec x}}\\
    \Aase_2&\eqdef\set{\scale n {\vec x}\vecplus\scale m {\vec y}}
             \cup\paren{\Aase\setminus\set{\vec y}}.
  \end{align*}
  We want to conclude by apply the induction hypothesis. To do so we
  must check that each of the linear expressions in the decomposition
  of $e$ are still $B$-homogeneous, and that their $\vec b$-score has
  decreased strictly.  For the first check, just notice that for
  $\vec a\in B$ since the sign of the $\vec a$-coordinates of $\vec x$
  and $\vec y$ is the same, the sign of
  $\scale n {\vec x}\vecplus\scale m {\vec y}$ is the same
  again. Therefore both $\Aase_1$ and $\Aase_2$ are
  $B$-homogeneous. For the second check, it follows immediately from
  the definitions that the $\vec b$-score of both $\Aase_1$ and
  $\Aase_2$ is one less than that of $\Aase$. We may thus conclude the
  proof by applying the induction hypothesis to each.\qed
\end{proof}


Finally, we show the main result of this section, namely:
\baseSplit
\begin{proof}
  Thanks to Lemma~\ref{lem:homogeneous} and
  Remark~\ref{rmk:dec-union}, it is enough to show that every
  homogeneous expression can be decomposed.  We do so by induction on
  the dimension of the expression.

  Let $e=\vec u\tadd\titerate{\Aase}$ be a homogeneous
  expression. We distinguish two cases: 
  \begin{enumerate}
  \item either $\forall \vec a\in \Aase$ we have
    $\topointsc{\vec a}{\vec b}$ is non-negative for every
    $\vec b\in\Base$ and $0$ otherwise,
  \item or there exists $\dv\in\Aase$ and $\vec b\in\BBase$ such that
    either
    \begin{orlist}
    \item\label{item:or:1} $\vec b\in \Base$ and
      $\topointsc{\dv}{\vec b}<0$, or
    \item\label{item:or:2} $\vec b\in {\bar\Base}$ and
      $\topointsc{\dv}{\vec b} < 0$, or
    \item\label{item:or:3} $\vec b\in {\bar\Base}$ and
      $\topointsc{\dv}{\vec b} > 0$.
    \end{orlist}
  \end{enumerate}
  Let us deal with each case in turn.
  \begin{enumerate}
  \item in this case, we show that $e$ can be written as a finite
    \union of expressions satisfying the premise of
    Lemma~\ref{lem:dec-compatible}, which allows us to conclude.  To
    do that, notice that for every vector $\vec a\in\Vectors$, there
    is a natural number $n_{\vec a}>0$ such that every coordinate of
    $\topoints{\scale {n_{\vec a}}{\vec a}}$ is an
    integer. Furthermore, if $\vec a\in\Aase$, then the $\bar\Base$
    coordinates of $\topoints{\scale {n_{\vec a}}{\vec a}}$ are equal
    to naught, and the $\Base$ coordinates of
    $\topoints{\scale{n_{\vec a}}{\vec a}}$ are natural numbers. This
    entails that $\scale {n_{\vec a}}{\vec a}\in\iterate\Base$. We may
    thus conclude this case using~\eqref{eq:first-iterates}:
    \ifcs
    \begin{align*}
      \vec u\tadd\titerate{ \Aase} 
      \equiv\vec u\tadd\tsum_{\vec a\in\Aase}\titerate{\vec a}
      \equiv\vec u\tadd\tsum_{\vec a\in\Aase}\paren{{\vec a}^{<n_{\vec a}}\tadd\titerate{\scale {n_{\vec a}}{\vec a}}}
      &\equiv\paren{\vec u\tadd\tsum_{\vec a\in\Aase}{\vec a}^{<n_{\vec a}}}
        \tadd\ptiterate{\tJoin_{\vec a\in\Aase}\scale {n_{\vec a}}{\vec a}}
    \end{align*}
    \else
    \begin{align*}
      \vec u\tadd\titerate{ \Aase} 
      \equiv\vec u\tadd\tsum_{\vec a\in\Aase}\titerate{\vec a}
      &\equiv\vec u\tadd\tsum_{\vec a\in\Aase}\paren{{\vec a}^{<n_{\vec a}}\tadd\titerate{\scale {n_{\vec a}}{\vec a}}}\\
      &\equiv\paren{\vec u\tadd\tsum_{\vec a\in\Aase}{\vec a}^{<n_{\vec a}}}
        \tadd\ptiterate{\tJoin_{\vec a\in\Aase}\scale {n_{\vec a}}{\vec a}}
    \end{align*}
    \fi
  \item this case as three sub-cases. Since all three can be dispatched
    in the same way, we only detail the proof in case~\ref{item:or:1},
    where we have $\dv\in\Aase$ and $\vec b\in \Base$ such that
    $\topointsc{\dv}{\vec b}<0$.
    Let $N=\left\lceil\topointsc{\vec u}{\vec b}\right\rceil+1$, and
    $\vec u'\eqdef\vec u\tadd\scale N\dv$. We rewrite $e$ as follows:
    \begin{align*}
      e=\vec u\tadd\titerate{\Aase}
      &\equiv\vec u\tadd\titerate\dv\tadd\ptiterate{\Aase\setminus\dv}\\
      &\equiv\vec u\tadd\paren{\dv^{<N}\tjoin\scale N\dv\tadd\titerate\dv}
        \tadd\ptiterate{\Aase\setminus\dv}\\
      &\equiv\vec u\tadd\dv^{<N}\tadd\ptiterate{\Aase\setminus\dv}
        \tjoin\vec u\tadd\scale N\dv
        \tadd\titerate\dv\tadd\ptiterate{\Aase\setminus\dv}\\
      &\equiv
        \begin{array}[t]{ccc}
          \paren{\underbrace{\vec u\tadd\dv^{<N}\tadd\ptiterate{\Aase\setminus\dv}}}
          &\tjoin&\paren{\underbrace{\vec u'\tadd\titerate{\Aase}}}.\\
          e'&&y
        \end{array}
    \end{align*}
    The expression $e'$ has dimension strictly smaller than
    $\Card\Aase$, so we can decompose it using the induction
    hypothesis.  We now show that $y$ does not intersect
    $\iterate\Base$, hence $\tuple{\zero,y}$ is a decomposition of
    $y$.
    Let $\vec v\in\vec u'\vecplus\iterate\Aase$. Since $e$ is
    homogeneous, and $\topointsc{\dv}{\vec b}<0$, every vector in
    $\Aase$ has non-positive $\vec b$-coordinates. It then follows
    that every vector in $\iterate\Aase$ has non-positive
    $\vec b$-coordinates. Therefore:
    \begin{align*}
      \topointsc{\vec v}{\vec b}
      =\topointsc{\vec v-\vec u'}{\vec b}+\topointsc{\vec u'}{\vec b}
      &\leqslant 0+\topointsc{\vec u'}{\vec b}
        \tag{$\vec v-\vec u'\in\iterate\Aase$}\\
      &=\topointsc{\vec u}{\vec b}+ N\times\topointsc{\dv}{\vec b}\\
      &\leqslant\topointsc{\vec u}{\vec b}+ N\times\paren{-1}=\topointsc{\vec u}{\vec b}- N.
    \end{align*}
    We know that $\topointsc{\vec u}{\vec b}<\left\lceil\topointsc{\vec u}{\vec b}\right\rceil+1=N$, so $\topointsc{\vec u}{\vec b}- N<0$, meaning $\topointsc{\vec v}{\vec b}<0$.
    If $\vec v$ were in $\iterate \Base$, then there would be a point
    $p\in\Points$ such that $\tovec p=\vec v$. By definition of
    $\topoints\_$ we have that
    $\topointsc{v}{\vec b}=\topointsc{\tovec p}{\vec b}=p(\vec
    b)\in\Nat$.  Since we have just showed that
    $\topointsc{v}{\vec b}<0$, this is impossible so
    $\vec v\notin\iterate\Base$.\qed
  \end{enumerate}
\end{proof}


\section*{Acknowledgements}

This work is supported by IRIS, UK EPSRC project EP/R006865/1.

%
%
%
%
\clearpage
\bibliographystyle{splncs04}
\bibliography{bibli}
\clearpage
\appendix
\section{Omitted proofs of Section~\ref{sec:plan}}
\label{sec:proofs:2.3}

\begin{lemma}
  For any $e\in\reg$ if $\sem e=\emptyset$, then $e\equiv \zero$.
\end{lemma}
\begin{proof}
  We show by induction on $e\in\reg$ that either
  $\exists \vec v\in\sem e$ or $e\equiv \zero$:
  \begin{description}
  \item[$\blacktriangleright e=\zero,\un,a$:] trivial.
  \item[$\blacktriangleright e=e_1\tadd e_2$:] if $e_1$ or $e_2$ is provably equal to
    $\zero$, then by \eqref{cka-seq-absorbing} we get $e\equiv \zero$;
    otherwise we have $\vec v_1\in\sem {e_1}$ and
    $\vec v_2\in\sem {e_2}$, hence $\vec v_1\vecplus\vec v_2\in\sem e$.
  \item[$\blacktriangleright e=e_1\tjoin e_2$:] if both $e_1$ and $e_2$ are provably equal to $\zero$, then thanks to \eqref{cka-plus-idem} $e\equiv \zero$; otherwise we have either $\vec v\in\sem {e_1}$ or $\vec v \in\sem{e_2}$, and in both cases we get $\vec v\in\sem e$;
  \item[$\blacktriangleright e=\titerate{f}$:] we have $\unitvec\in\sem e$.\qed
  \end{description}
\end{proof}

\section{Omitted proofs of Section~\ref{sec:finite-complete}}
\label{sec:proofs:3.1}

\splitfinite*
\begin{proof}
  By induction on $e$:
  \begin{description}
  \item[$\blacktriangleright \zero,\un,a$:] In each of those case, we have $e= \tJoin_{\vec v\in\sem e}\expr{\vec v}$, therefore the lemma holds by reflexivity.
  \item[$\blacktriangleright f\tjoin g$:]
    \begin{align*}
      f\tjoin g
      &\equiv \jparen{\tJoin_{\vec v\in\sem f}\expr{\vec v}}
        \tjoin\jparen{\tJoin_{\vec v\in\sem g}\expr{\vec v}}\tag{by I.H.}\\
      &\equiv \tJoin_{\vec v\in\sem f\join\sem g}\expr{\vec v}.
        \tag{by~\ref{cka-plus-ass}}
    \end{align*}
  \item[$\blacktriangleright f\tadd g$:]
    \begin{align*}
      f\tadd g
      &\equiv \jparen{\tJoin_{\vec u\in\sem f}\expr{\vec u}}
        \tadd\jparen{\tJoin_{\vec v\in\sem g}\expr{\vec v}}
        \tag{by I.H.}\\
      &\equiv 
        \tJoin_{\vec u \in\sem f,\vec v\in\sem g}\expr{\vec u}\tadd\expr{\vec v}
        \tag{by~\ref{cka-left-distr},\ref{cka-seq-comm}}\\
      &\equiv 
        \tJoin_{\vec u \in\sem f,\vec v\in\sem g}\expr{\vec u\vecplus\vec v}
        \tag{by~Lemma~\ref{lem:add-vect}}\\
      &\equiv
        \tJoin_{\vec w\in\sem {f\tadd g}}\expr{\vec w}\tag*{\qed}
    \end{align*}
  \end{description}
\end{proof}

\section{Omitted proofs of Section~\ref{sec:stuff}}
\label{sec:proofs:3.2}

\begin{lemma}
  For any $k\in\Nat$ we have $e^{<k+1}\equiv \scale k e\tjoin e^{<k}$.
\end{lemma}
\begin{proof}
  First, notice that the right to left inequality holds directly:
  \begin{align*}
    \scale k e\tjoin e^{<k}
    =\scale k e\tjoin\scale {\paren{k-1}}{\paren{e\tjoin\un}}
    &\equiv\scale k e\tjoin\jparen{\scale {\paren{k-1}}{\paren{e\tjoin\un}}\tadd\un}\\
    &\leqq \scale k {\paren{e\tjoin \un}}\tjoin\jparen{\scale {\paren{k-1}}{\paren{e\tjoin\un}}\tadd\paren{e\tjoin \un}}\\
    &\equiv \scale k{\paren{e\tjoin \un}}= e^{<k+1}.
  \end{align*}

  We now prove the conversion inequality by induction on $k$.
  \begin{description}
  \item[case $k=0$:] we have
    $$e^{<k+1}=\scale 0 {\paren{e\tjoin \un}}=\un\leqq \un\tjoin e^{<0}=\scale 0 e\tjoin e^{<0}.$$
  \item[case $k+1$:] in this case, we have
    \ifcs
    \begin{align*}
      e^{<k+2}=\scale {\paren{k+1}}{\paren{e\tjoin\un}}
      &=\paren{e\tjoin \un}\tadd e^{<k+1}\\
      &\equiv e\tadd e^{<k+1}\tjoin e^{<k+1}\\
      &\leqq e\tadd \paren{\scale k e \tjoin e^{<k}}\tjoin e^{<k+1}\tag{by I.H.}\\
      &\equiv\scale {\paren{k+1}} e\tjoin \jparen{e\tadd e^{<k}}\tjoin e^{<k+1}\\
      &\leqq\scale {\paren{k+1}} e\tjoin \jparen{\paren{e\tjoin \un}\tadd e^{<k}}\tjoin e^{<k+1}\\
      &\equiv\scale {\paren{k+1}} e\tjoin e^{<k+1}\tjoin e^{<k+1}\equiv \scale  {\paren{k+1}} e\tjoin e^{<k+1}.\tag*{\qed}
    \end{align*}
    \else
    \begin{align*}
      e^{<k+2}=\scale {\paren{k+1}}{\paren{e\tjoin\un}}
      &=\paren{e\tjoin \un}\tadd e^{<k+1}\\
      &\equiv e\tadd e^{<k+1}\tjoin e^{<k+1}\\
      &\leqq e\tadd \paren{\scale k e \tjoin e^{<k}}\tjoin e^{<k+1}\tag{by I.H.}\\
      &\equiv\scale {\paren{k+1}} e\tjoin \jparen{e\tadd e^{<k}}\tjoin e^{<k+1}\\
      &\leqq\scale {\paren{k+1}} e\tjoin \jparen{\paren{e\tjoin \un}\tadd e^{<k}}\tjoin e^{<k+1}\\
      &\equiv\scale {\paren{k+1}} e\tjoin e^{<k+1}\tjoin e^{<k+1}\\
      &\equiv \scale  {\paren{k+1}} e\tjoin e^{<k+1}.\tag*{\qed}
    \end{align*}
    \fi
  \end{description}
\end{proof}
\begin{lemma}
\begin{equation*}
  \ptiterate{e\tjoin f}\equiv\titerate e\tadd\titerate f\tag{\ref{eq:sum-to-prod}}
\end{equation*}
\end{lemma}
\begin{proof}
  Since $e\leqq e\tjoin f$ and $f\leqq e\tjoin f$, and since $\leqq$ is
  a precongruence, we get $\titerate e\leqq \ptiterate{e\tjoin f}$ and
  $\titerate f\leqq \ptiterate{e\tjoin f}$, hence:
  $$\titerate e\tadd\titerate f\leqq\ptiterate{e\tjoin
    f}\tadd\ptiterate{e\tjoin f}\equiv \ptiterate{e\tjoin f}.$$

  For the converse direction, we start by showing the following inequality:
  \begin{equation}
    \paren{e\tjoin f}\tadd\paren{\titerate e\tadd\titerate f}\leqq \titerate e\tadd\titerate f.\eqproof{1}
  \end{equation}
  \begin{align*}
    \paren{e\tjoin f}\tadd\jparen{\titerate e\tadd\titerate f}
    &\equiv \jparen{e\tadd\jparen{\titerate e\tadd\titerate f}}
      \tjoin\jparen{f\tadd\jparen{\titerate e\tadd\titerate f}}\\
    &\equiv\jparen{\jparen{e\tadd\titerate e}\tadd\titerate f}
      \tjoin\jparen{\titerate e\tadd\jparen{f\tadd\titerate f}}\\
    &\leqq\jparen{\titerate e\tadd\titerate f}
      \tjoin\jparen{\titerate e\tadd\titerate f}\equiv\titerate e\tadd\titerate f.
  \end{align*}
  By~\eqref{cka-lfp}, the inequality~\eqref{eq:1} entails that
  $\ptiterate{e\tjoin f}\tadd\paren{\titerate e\tadd\titerate f}\leqq \titerate e\tadd\titerate f.$
  We may thus conclude, because since $\un\leqq\titerate e\tadd\titerate f$ we have:
  $$\ptiterate {e\tjoin f}\equiv\ptiterate {e\tjoin f}\tadd\un\leqq\ptiterate{e\tjoin f}\tadd\jparen{\titerate e\tadd\titerate f}\leqq \titerate e\tadd\titerate f.$$

  \noindent%
  We have therefore proved both inequalities, we may conclude by
  antisymmetry.\qed
\end{proof}

\basevector*
\begin{proof}
  The inclusion from right to left is trivial, so we focus on the
  other one. By~\eqref{cka-lfp}, and given\footnotemark~that the right-hand side is
  provably larger than $\un$, this amounts to proving:
  \footnotetext{For this to hold it is necessary to check that
    $p\neq\unitvec$, i.e. $\exists i,\,p_i\neq 0$.}
  \begin{align*}
    &\Base\tadd\titerate{\tovec p}\tadd\tJoin_{i}\paren{\vec u_i^{<p_i}\tadd\tsum_{i\neq j}\titerate{\vec u_j}}\leqq\titerate{\tovec p}\tadd\tJoin_{i}\paren{\vec u_i^{<p_i}\tadd\tsum_{i\neq j}\titerate{\vec u_j}}\\
    \Leftrightarrow~&\forall i,k:\,
                      \vec u_k\tadd\titerate{\tovec p}\tadd\vec u_i^{<p_i}\tadd\tsum_{i\neq j}\titerate{\vec u_j}\leqq\titerate{\tovec p}\tadd\tJoin_{i}\paren{\vec u_i^{<p_i}\tadd\tsum_{i\neq j}\titerate{\vec u_j}}
  \end{align*}
  We consider two cases, depending on $i\stackrel{?}{=}k$:

  \noindent%
  $\blacktriangleright i\neq k$\textbf{:} In this case, we have $\vec u_k\in\setcompr{\vec u_j}{j\neq i}$, hence
  $\vec u_k\tadd\tsum_{i\neq j}\titerate{\vec u_j}\leqq\tsum_{i\neq j}\titerate{\vec u_j}$, so we obtain:
  \begin{align*}
    \vec u_k\tadd\titerate{\tovec p}\tadd\vec u_i^{<p_i}\tadd\tsum_{i\neq j}\titerate{\vec u_j}
    &\leqq\titerate{\tovec p}\tadd\vec u_i^{<p_i}\tadd\tsum_{i\neq j}\titerate{\vec u_j}\\
    &\leqq\titerate{\tovec p}\tadd\tJoin_{i}\paren{\vec u_i^{<p_i}\tadd\tsum_{i\neq j}\titerate{\vec u_j}}.
  \end{align*}

  \noindent%
  $\blacktriangleright i=k$\textbf{:} In this case, we observe that since
  $\vec u_i\tadd\vec u_i^{<n}\leqq \vec
  u_i^{<n}\tjoin\scale n {\vec u_i}$, we have:
  \ifcs
  \begin{align*}
    &\vec u_i\tadd\titerate{\tovec p}\tadd\vec u_i^{<p_i}\tadd\tsum_{i\neq j}\titerate{\vec u_j}
      \leqq
      \titerate{\tovec p}\tadd\vec u_i^{<p_i}\tadd\tsum_{i\neq j}\titerate{\vec u_j}\tjoin
      \titerate{\tovec p}\tadd\scale{p_i}{\vec u_i}\tadd\tsum_{i\neq j}\titerate{\vec u_j}
  \end{align*}
  \else
  \begin{align*}
    &\vec u_i\tadd\titerate{\tovec p}\tadd\vec u_i^{<p_i}\tadd\tsum_{i\neq j}\titerate{\vec u_j}
    \\ \leqq~&
               \titerate{\tovec p}\tadd\vec u_i^{<p_i}\tadd\tsum_{i\neq j}\titerate{\vec u_j}\tjoin
               \titerate{\tovec p}\tadd\scale{p_i}{\vec u_i}\tadd\tsum_{i\neq j}\titerate{\vec u_j}
  \end{align*}
  \fi
  \noindent%
  Since the first term is smaller than our goal, we focus on the second one:
  \begin{align*}
    \ifcs \else & \fi
                  \titerate{\tovec p}\tadd\scale{p_i}{\vec u_i}\tadd\tsum_{i\neq j}\titerate{\vec u_j}
                  \ifcs \else \\ \fi
    \equiv~&
             \titerate{\tovec p}\tadd\scale{p_i}{\vec u_i}\tadd\tsum_{i\neq j}\paren{\vec u_j^{<p_j}\tjoin\jparen{\scale{p_j}{\vec u_j}\tadd\titerate{\vec u_j}}}
             \tag{by~\ref{eq:n-or-more}}\\
    \equiv~& \tJoin_{I\subseteq\set{1,\dots,n}\setminus i}
             \titerate{\tovec p}\tadd\scale{p_i}{\vec u_i}\tadd\tsum_{\mbox{\scriptsize
             $\begin{array}{c}
               j\neq i\\j\notin I
             \end{array}$}
    }\vec u_j^{<p_j}\tadd\tsum_{j\in I}\paren{\scale{p_j}{\vec u_j}\tadd\titerate{\vec u_j}}\\
    \equiv~& \tJoin_{I\subsetneq\set{1,\dots,n}\setminus i}
             \titerate{\tovec p}\tadd\scale{p_i}{\vec u_i}\tadd\tsum_{\mbox{\scriptsize
             $\begin{array}{c}
               j\neq i\\j\notin I
             \end{array}$}
    }\vec u_j^{<p_j}\tadd\tsum_{j\in I}\scale{p_j}{\vec u_j}\tadd\tsum_{j\in I}\titerate{\vec u_j} \\
                &\tjoin  \titerate{\tovec p}\tadd\scale{p_i}{\vec u_i}\tadd\tsum_{i\neq j}\scale{p_j}{\vec u_j}\tadd\tsum_{i\neq j}\titerate{\vec u_j}
  \end{align*}
  We consider those two terms separately:
  \begin{itemize}
  \item For the first term, consider
    $I\subsetneq\set{1,\dots,n}\setminus i$.  Since
    $I\neq\set{1,\dots,n}\setminus i$, there is an index $k$ such that
    $k\neq i$ and $k\notin I$.  Observe that we get:
    \begin{mathpar}
      \forall j\in I,\,u_j\leqq \tsum_{j\neq k}\titerate{\vec u_j}
      \and
      \scale{p_i}{\vec u_i}\leqq \tsum_{j\neq k}\titerate{\vec u_j}
      \and
      \tsum_{\mbox{\scriptsize$
        \begin{array}{c}
          j\neq i\\j\notin I
        \end{array}$}}\vec u_j^{<p_j}\equiv\vec u_k^{<p_k}\tadd\tsum_{\mbox{\scriptsize$
        \begin{array}{c}
          j\neq i,k\\j\notin I
        \end{array}$}}\vec u_j^{<p_j}\leqq\vec u_k^{<p_k}\tadd\tsum_{j\neq k}\titerate{\vec u_j}.
    \end{mathpar}
    Therefore we obtain
    \begin{align*}
      &\titerate{\tovec p}\tadd\scale{p_i}{\vec u_i}\tadd\tsum_{\mbox{\scriptsize$
        \begin{array}{c}
          j\neq i\\j\notin I
        \end{array}$}}\vec u_j^{<p_j}\tadd\tsum_{j\in I}\scale{p_j}{\vec u_j}\tadd\tsum_{j\in I}\titerate{\vec u_j}\\
      \leqq~&
              \titerate{\tovec p}\tadd   \tsum_{k\neq j}\titerate{\vec u_j}
              \tadd\vec u_k^{<p_k}\tadd\tsum_{k\neq j}\titerate{\vec u_j}
              \tadd\tsum_{k\neq j}\titerate{\vec u_j}
              \tadd\tsum_{k\neq j}\titerate{\vec u_j}\\
      \equiv~&\titerate{\tovec p}
               \tadd\vec u_k^{<p_k}
               \tadd\tsum_{k\neq j}\titerate{\vec u_j}
               \leqq\titerate{\tovec p}\tadd
               \tJoin_{k}
               \paren{\vec u_k^{<p_k}
               \tadd\tsum_{k\neq j}\titerate{\vec u_j}}.
    \end{align*}
  \item for the second term, i.e.
    $\titerate{\tovec p}\tadd\scale{p_i}{\vec u_i}\tadd\tsum_{i\neq
      j}\scale{p_j}{\vec u_j}\tadd\tsum_{i\neq j}\titerate{\vec u_j}$,
    first we notice:
    $$\scale{p_i}{\vec u_i}\tadd\tsum_{i\neq
      j}\scale{p_j}{\vec u_j}\equiv\tsum_j\scale{p_j}{\vec u_j} \equiv\tovec p.$$
    Therefore we have
    \begin{align*}
      &\titerate{\tovec p}\tadd\scale{p_i}{\vec u_i}\tadd\tsum_{i\neq
        j}\scale{p_j}{\vec u_j}\tadd\tsum_{i\neq j}\titerate{\vec u_j}
        \leqq
        \titerate{\tovec p}\tadd\tovec p\tadd\tsum_{i\neq j}\titerate{\vec u_j}\\
      \leqq~&
              \titerate{\tovec p}\tadd\tsum_{i\neq j}\titerate{\vec u_j}
              \leqq\titerate{\tovec p}\tadd
              \vec u_i^{<p_i}
              \tadd\tsum_{i\neq j}\titerate{\vec u_j}\\
      \leqq~&\titerate{\tovec p}\tadd
              \tJoin_{k}
              \paren{\vec u_k^{<p_k}
              \tadd\tsum_{k\neq j}\titerate{\vec u_j}}.\tag*{\qed}
    \end{align*}
  \end{itemize}
\end{proof}

\section{Omitted proofs of Section~\ref{sec:rat}}
\label{sec:proofs:3.3}

\linearindep*
\begin{proof}
  \begin{description}
  \item[($\Rightarrow$)] Suppose $\Base$ is independent. By our
    definition this means that $\tovec\_$ is injective, i.e. for any
    pair of $\Base$-points $\alpha,\beta\in\Points$,
    $\tovec\alpha=\tovec\beta\Rightarrow \alpha=\beta$.

    Let $p\in\Rat^\Base$ be a rational $\Base$-point such that
    $$\vsum_{\vec u\in\Base}\scale{p\paren{\vec u}}{\vec u}=\unitvec.$$
    First, we use the fact that since
    $\setcompr{p\paren{\vec u}}{\vec u\in\Base}$ is a finite set of
    rational numbers, there exists a natural number $N$ such that for
    any $\vec u\in \Base$ the number $N\times p\paren{\vec u}$ is an
    integer.
    We now define two points $\alpha,\beta\in\Points$:
    \begin{align*}
      \alpha&\eqdef\brack{\vec u\mapsto\left\{
              \begin{array}{ll}
                N\times p\paren{\vec u}&\text{ if }p\paren{\vec u}>0\\
                0&\text{ otherwise}\\
              \end{array}
      \right.}\\
      \beta&\eqdef\brack{\vec u\mapsto\left\{
             \begin{array}{ll}
               -N\times p\paren{\vec u}&\text{ if }p\paren{\vec u}\leqslant 0\\
               0&\text{ otherwise}\\
             \end{array}
      \right.}
    \end{align*}
    It is now a simple exercise to check that
    $\tovec\alpha=\tovec\beta$, which means that $\alpha=\beta$. By
    unfolding the definitions, this implies that
    $\alpha=\beta=p=\unitvec$.
  \item[($\Leftarrow$)] Now assume that $\Base$ is linearly
    independent, and let $\alpha,\beta\in\Points$ such that
    $\tovec \alpha=\tovec \beta$. Now, if we define $p=\alpha-\beta$,
    we get that:
    $$\vsum_{\vec u\in\Base}\scale{p\paren{\vec u}}{\vec u}=\vsum_{\vec u\in\Base}\scale{\alpha\paren{\vec u}}{\vec u}-\vsum_{\vec u\in\Base}\scale{\beta\paren{\vec u}}{\vec u}=\tovec \alpha-\tovec \beta=\unitvec.$$
    Since $\Base$ is linearly independent, $p$ must be uniformly zero,
    i.e. $\alpha=\beta$.\qed
  \end{description}
\end{proof}

\end{document}

%% file: packages.tex
\usepackage{mathpartir}
\usepackage{enumitem}
\usepackage[bbgreekl]{mathbbol}
\usepackage{bbm}
\usepackage{amssymb}
\usepackage{tikz}
\usetikzlibrary{arrows,fit,calc}
\usepackage{array}
\usepackage{scalerel}
\usepackage{xparse}
\usepackage{thm-restate}
\usepackage{xspace}

\usepackage[pdftex,
plainpages=false,
pdfpagelabels,
linktoc=all,
bookmarks=true,
bookmarksopen=true,
breaklinks=true,
colorlinks,
linkcolor=black,
urlcolor=black,
citecolor=black]{hyperref}


%% file: macros-cs.tex
\newlist{deflist}{enumerate*}{1}
\setlist[deflist]{label=(\roman*)}
\newlist{orlist}{enumerate*}{1}
\setlist[orlist]{label=(\alph*)}

\DeclareSymbolFont{stmry}{U}{stmry}{m}{n}
\DeclareMathDelimiter\llbracket{\mathopen}{stmry}{"4A}{stmry}{"71}
\DeclareMathDelimiter\rrbracket{\mathclose}{stmry}{"4B}{stmry}{"79}
\DeclareMathSymbol\llparenthesis\mathopen{stmry}{"4C}
\DeclareMathSymbol\rrparenthesis\mathclose{stmry}{"4D}
\DeclareMathSymbol\subsetplus\mathrel{stmry}{"44}
\DeclareMathSymbol\supsetplus\mathrel{stmry}{"45}

\DeclareMathAlphabet{\mathpzc}{OT1}{pzc}{m}{it}

\newcommand\commutative{}
\newcommand\union{sum\xspace}
\newcommand\unions{sums\xspace}

\newcommand\paren[1]{\left(#1\right)}
\renewcommand\brack[1]{\left[#1\right]}
\newcommand\set[1]{\left\{#1\right\}}
\newcommand\setcompr[2]{\left\{#1~\middle|~#2\right\}}
\newcommand\tuple[1]{\left\langle#1\right\rangle}

\newcommand\Mid{\mathrel{|}}
\newcommand\Coloneqq{\mathrel{\mathop{::}}=}
\newcommand\eqdef{\mathrel{\mathop{:}}=}

\newcommand\Nat{\mathbb{N}}
\newcommand\Rat{\mathbb{Q}}

\renewcommand\epsilon\varepsilon

\newcommand\sem[1]{\left\llbracket {#1}\right\rrbracket}

\newcommand\alphabet{\mathbb\Sigma}
\newcommand\Card[1]{{\#{#1}}}

\newcommand\reg[1][\alphabet]{\mathrm{Reg}_{#1}}
\newcommand\un{\mathbf 1}
\newcommand\zero{\mathbf 0}
\newcommand\tadd{\cdot}
\newcommand\tjoin{+}

\newcommand\tJoin{\sum}
\newcommand\titerate[1]{{#1}^\star}
\newcommand\ptiterate[1]{{\paren{#1}}^\star}

\newcommand\expr[1]{\brack{#1}_{\mathpzc{e}}}

\renewcommand\dim[1]{\text{dim}{\paren{#1}}}
\DeclareDocumentCommand\restre{ O{e} O{\Gamma} }{{#1}_{#2}}
\DeclareDocumentCommand\restreb{ O{e} O{\Gamma} }{\restre[#1][\bar{#2}]}
\newcommand\tsum{\prod}

\newcommand\jparen[1]{#1}

\DeclareDocumentCommand\Vectors{ O{\alphabet} D(){\Nat} }{{#2}^{#1}}
\newcommand\Base{\mathcal{B}}
\newcommand\BBase{\mathbf{B}}
\newcommand\Aase{\mathcal{A}}
\newcommand\Points[1][\Base]{{\mathbb{P}}{\tuple{#1}}}
\newcommand\tovec[2][\Base]{\left\langle\hspace{-2px}\middle|#2\middle|\hspace{-2px}\right\rangle_{#1}}
\newcommand\topoints[2][\BBase]{\brack{#2}^{#1}}
\newcommand\topointsc[3][\BBase]{\brack{#2}^{#1}_{#3}}
\newcommand\vu[1][]{\mathpzc{u}_{#1}}
\newcommand\vv[1][]{\mathpzc{v}_{#1}}
\newcommand\vU[1][]{\mathpzc{U}_{#1}}
\newcommand\vV[1][]{\mathpzc{V}_{#1}}
\newcommand\vecplus{\oplus}
\newcommand\scale[2]{{#2}^{#1}}
\newcommand\ituple[2]{\tuple{#1}_{#2}}
\newcommand\join{\cup}
\renewcommand\Join{\bigcup}
\newcommand\iterate[1]{{#1}^\star}

\renewcommand\vec[1]{\mathpzc{#1}}
\newcommand\unitvec{\epsilon}
\newcommand\letvec[1]{\left\{\hspace{-2.5px}\middle|#1\middle|\hspace{-2.5px}\right\}}
\newcommand\dv{\mathfrak{d}}
\newcommand\vsum{\bigoplus}

\newcounter{mysubtable}
\newcounter{myaxiom}
\newcounter{mylaw}
\newcounter{myeq}
\newcommand\axlabel[1]{%
  \stepcounter{myaxiom}%
  \label{#1}%
  \tag{\Alph{mysubtable}{\scriptsize{\arabic{myaxiom}}}}%
}

\newcommand\law[1]{%
  \stepcounter{mylaw}%
  \label{eq:#1}%
  \tag{E{\scriptsize{\arabic{mylaw}}}}%
}

\newcommand\eqproof[1]{%
  \stepcounter{myeq}%
  \label{eq:#1}%
  \tag{\fnsymbol{myeq}}%
}

\newenvironment{tableequations}{%
  \setcounter{myaxiom}{0}%
  \stepcounter{mysubtable}%
  \ignorespaces
}{%
  \ignorespacesafterend
}
\numberwithin{equation}{section}
